\documentclass{IEEEtran}
\usepackage{times,amssymb,amsmath,epsfig,nicefrac,euscript,cite,mathrsfs,balance,mathrsfs}

\newcommand\nc\newcommand
\newcommand{\Author}[1]{\text{#1}}
  
\nc\bfa{{\boldsymbol a}}\nc\bfA{{\boldsymbol A}}\nc\cA{{\mathcal A}}
\nc\bfb{{\boldsymbol b}}\nc\bfB{{\boldsymbol B}}\nc\cB{{\mathcal B}}
\nc\bfc{{\boldsymbol c}}\nc\bfC{{\boldsymbol C}}\nc\cC{{\mathcal C}}
\nc\sC{{\mathscr C}}
\nc\bfd{{\boldsymbol d}}\nc\bfD{{\boldsymbol D}}\nc\cD{{\mathcal D}}
\nc\bfe{{\boldsymbol e}}\nc\bfE{{\boldsymbol E}}\nc\cE{{\mathcal E}}
\nc\bff{{\boldsymbol f}}\nc\bfF{{\boldsymbol F}}\nc\cF{{\mathcal F}}
\nc\bfg{{\boldsymbol g}}\nc\bfG{{\boldsymbol G}}\nc\cG{{\mathcal G}}
\nc\bfh{{\boldsymbol h}}\nc\bfH{{\boldsymbol H}}\nc\cH{{\mathcal H}}
\nc\bfi{{\boldsymbol i}}\nc\bfI{{\boldsymbol I}}\nc\cI{{\mathcal I}}
\nc\bfj{{\boldsymbol j}}\nc\bfJ{{\boldsymbol J}}\nc\cJ{{\mathcal J}}
\nc\bfk{{\boldsymbol k}}\nc\bfK{{\boldsymbol K}}\nc\cK{{\mathcal K}}
\nc\bfl{{\boldsymbol l}}\nc\bfL{{\boldsymbol L}}\nc\cL{{\mathcal L}}
\nc\bfm{{\boldsymbol m}}\nc\bfM{{\boldsymbol M}}\nc\sM{{\mathscr M}}
\nc\bfn{{\boldsymbol n}}\nc\bfN{{\boldsymbol N}}\nc\cN{{\mathcal N}}
\nc\bfo{{\boldsymbol o}}\nc\bfO{{\boldsymbol O}}\nc\cO{{\mathcal O}}
\nc\bfp{{\boldsymbol p}}\nc\bfP{{\boldsymbol P}}\nc\cP{{\mathcal P}}
\nc\bfq{{\boldsymbol q}}\nc\bfQ{{\boldsymbol Q}}\nc\cQ{{\mathcal Q}}
\nc\bfr{{\boldsymbol r}}\nc\bfR{{\boldsymbol R}}\nc\cR{{\mathcal R}}
\nc\bfs{{\boldsymbol s}}\nc\bfS{{\boldsymbol S}}\nc\cS{{\mathcal S}}
\nc\bft{{\boldsymbol t}}\nc\bfT{{\boldsymbol T}}\nc\cT{{\mathcal T}}
\nc\bfu{{\boldsymbol u}}\nc\bfU{{\boldsymbol U}}\nc\cU{{\mathcal U}}
\nc\bfv{{\boldsymbol v}}\nc\bfV{{\boldsymbol V}}\nc\cV{{\mathcal V}}
\nc\bfw{{\boldsymbol w}}\nc\bfW{{\boldsymbol W}}\nc\cW{{\mathcal W}}
\nc\bfx{{\boldsymbol x}}\nc\bfX{{\boldsymbol X}}\nc\cX{{\mathcal X}}
\nc\bfy{{\boldsymbol y}}\nc\bfY{{\boldsymbol Y}}\nc\cY{{\mathcal Y}}
\nc\bfz{{\boldsymbol z}}\nc\bfZ{{\boldsymbol Z}}\nc\cZ{{\mathcal Z}}

\nc{\remove}[1]{}
\nc\nd\noindent
\nc\half{\nicefrac12}

\nc\uz{\underline z}

\def\dist{{d}}

\newtheorem{theorem}{Theorem}

\newtheorem{lemma}[theorem]{Lemma}
\newtheorem{proposition}[theorem]{Proposition}

\newcommand\ff{{\mathbb F}}

\begin{document}
\title{Constructions of Rank Modulation Codes}
\author{Arya Mazumdar$^{\ast}$,\;
{Alexander Barg$^{\S}$}~and~
{Gilles Z{\'e}mor$^{a}$}}
\maketitle

{\renewcommand{\thefootnote}{}\footnotetext{

\vspace{.02in}\nd The work of the first two authors
was supported by NSF grants CCF0916919, CCF0830699,
and DMS0807411.

\nd$^\ast$Department of EECS and Research Laboratory of Electronics,
Massachusetts Institute of Technology, Cambridge, MA, 02139, USA, e-mail: aryam@mit.edu.
Research done while at the Department of ECE and Institute for Systems Research,
University of Maryland, College Park, 20742, USA.

\nd$^\S$Department of ECE and Institute for Systems Research,
University of Maryland, College Park, 20742, USA, e-mail: abarg@umd.edu,
and
Dobrushin Mathematical Laboratory, Institute for Problems of
Information Transmission, Russian Academy of Science, Moscow, Russia.

\nd$^a$ Bordeaux Mathematics Institute, UMR 5251,
University of Bordeaux, France, e-mail:
 zemor@math.u-bordeaux1.fr.
}}
\renewcommand{\thefootnote}{\arabic{footnote}}
\setcounter{footnote}{0}

\maketitle
\begin{abstract}
Rank modulation is a way of encoding information to correct errors in
flash memory devices as well as impulse noise in transmission
lines. Modeling rank modulation involves construction of packings of
the space of permutations equipped with the Kendall tau distance.

We present several general constructions of codes in permutations that
cover a broad range of code parameters. In particular, we show a number of
ways in which conventional error-correcting codes can be modified to correct
errors in the Kendall space. Codes that we construct afford simple encoding
and decoding algorithms of essentially the same complexity as required to
correct errors in the Hamming metric. For instance, from binary BCH codes we obtain
codes correcting $t$ Kendall errors in $n$ memory cells that support
the order of $n!/(\log_2n!)^t$ messages, for any constant $t= 1,2,\dots.$ We also
construct families of codes that correct a number of errors that grows
with $n$ at varying rates, from $\Theta(n)$ to $\Theta(n^{2})$.
One of our constructions gives rise to a family of rank
modulation codes
for which the trade-off between the number of messages and the number of
correctable Kendall errors approaches the optimal scaling rate.
Finally, we list a number of possibilities for constructing codes of finite
length, and give examples of rank modulation codes with specific parameters.
\end{abstract}


\section{Introduction}\label{sec:intro}
Recently considerable attention in the literature was devoted to coding
problems for non-volatile memory devices, including error correction
in various models as well as data management in memories
\cite{jia10a,jia09b,jia09a,BM2010,cas10}.  Non-volatile memories, in
particular flash memory devices, store data by injecting charges of
varying levels in memory cells that form the device.
The current technology supports multi-level cells with two or more charge levels.
The write procedure into the memory is asymmetric in that it is possible to increase
the charge of an individual cell, while to decrease the charge one must
erase and overwrite a large block of cells using a mechanism called block erasure.
This raises the issue of data management in memory, requiring
data encoding for efficient rewriting of the data \cite{jia10b}.
A related issue concerns the reliability of the stored information which
is affected by the drift of the charge of the cells caused by ageing devices
or other reasons.  Since the drift in different cells may occur at
different speed, errors introduced in the data are adequately
accounted for by tracking the relative value of adjacent cells rather
than the absolute values of cell charges.  Storing information in relative
values of the charges also simplifies the rewriting of the data because
we do not need to reach any particular value of the charge as long as we have
the desired ranking, thereby reducing the risk of overprogramming.
Based on these ideas, Jiang
et al. \cite{jia10a,jia09b}
suggested to use the {\em rank modulation scheme} for error-correcting coding of
data in flash memories.  A similar noise model
arises in transmission over channels subject to impulse noise that
changes the value of the signal substantially but has less effect on
the relative magnitude of the neighboring signals. In an earlier work
devoted to modeling impulse noise, Chadwick and Kurz \cite{CK1969}
introduced the same error model and considered coding problems for
rank modulation. Drift of resistance in memory cells is also the
main source of errors in multilevel-cell phase-change memories \cite{pap11}.

Motivated by the application to flash memories, we consider reliable storage
of information in the rank modulation scheme. Relative ranks of cell charges
in a block of $n$ cells define a permutation on the set of $n$ elements.
Our problem therefore can be formulated as encoding of data into permutations
so that it can be recovered from errors introduced by the drift (decrease) of
the cell charges.

To define the error process formally, let $[n]=\{1,2\dots,n\}$ be a set of $n$
elements and consider the set $S_n$ of permutations of $[n].$
In this paper we use a one-line notation for
permutations: for instance (2,1,3) refers to the permutation
$\genfrac{(}{)}{0pt}{}{123}{213}.$ Referring to the discussion of charge levels of cells,
permutation $(2,1,3)$ means that the highest-charged cell is the second one
followed by the first and then the third cell.
Permutations can be multiplied by applying them successively to the set $[n]$,
namely the action of the permutation $\pi\sigma,$ where $\pi,\sigma\in S_n,$ results in
$i\mapsto \sigma(\pi(i)), i=1,\dots, n.$
(Here and elsewhere we assume that permutations act on the right).
Every permutation has an inverse, denoted $\sigma^{-1}$, and $e$ denotes the
identity permutation.

Let $\sigma=(\sigma(1),\dots,\sigma(n))$ be a permutation of $[n].$
An elementary error occurs when the charge of cell $j$ passes the level
of the charge of the cell with rank one smaller than the rank of $j$. If the $n$-block is
encoded into a permutation, $\sigma,$ this error corresponds to the exchanging
of the locations of the elements $\sigma(j)$ and $\sigma(j+1)$ in the permutation.
For instance, let $\sigma=(3,1,4,2)$ then the effect of the error $\pi=(2,1,3,4)$
is to exchange the locations of the two highest-ranked elements, i.e.,
$\pi\sigma =(1,3,4,2).$

Accordingly, define the {\em Kendall tau  distance}
$d_\tau(\sigma,\pi)$ from $\sigma$ to another permutation
$\pi$  as the minimum number of transpositions of pairwise
adjacent elements required to change $\sigma$ into $\pi.$
Denote by $\cX_n=(S_n,d_\tau)$ the metric space of permutations on $n$ elements
equipped with the distance $d_\tau.$
The Kendall metric was studied in statistics \cite{ken90}
where it was introduced as a measure of proximity of data samples,
as well as in combinatorics and coding theory \cite{DG1977,BM2010}.
The Kendall metric also arises naturally as a Cayley metric on the group $S_n$
if one takes the adjacent transpositions as its generators.

The Kendall distance is one of many metrics on permutations considered in the literature;
see the survey \cite{dez98}.
Coding for the Hamming metric was considered recently in \cite{chu04}
following the observation in \cite{vin00} that permutation arrays are useful
for error correction in powerline communication.
Papers \cite{klo10,shi10,tam10} considered
coding for the $\ell_\infty$ distance on permutations from the perspective of the
rank modulation scheme. Generalizations of Gray codes for rank modulation are
considered in \cite{yeh11}, while an application of LDPC codes to this scheme
is proposed in \cite{zha10}.

An $(n,d)$ code $\cC\subset \cX_n$ is a set of
permutations in $S_n$ such that the minimum distance $d_\tau$ separating any two of
them is at least $d.$
The main questions associated with the coding problem for the Kendall space $\cX_n$
are to establish the size of optimal codes
that correct a given number of errors and, more importantly, to construct explicit
coding schemes.
In our previous work \cite{BM2010} we addressed the first of these problems,
analyzing
both the finite-length and the asymptotic bounds on codes.
Since the maximum value
of the distance in $\cX_n$ is $\binom n2$, this leaves a number
of possibilities for the scaling rate of the distance for asymptotic
analysis, ranging from $d=O(n)$ to $d=\Theta(n^2).$ Define the rate of the code
\begin{equation}\label{eq:rank_rate}
R(\cC)=\log|\cC|/\log(n!)
\end{equation}
(all logarithms are base 2 unless otherwise mentioned) and let
  \begin{align}
    R(n,d)&=\max_{\cC\subset \cX_n} R(\cC) \label{eq:rnd}\\
   \sC(d)&=\lim_{n\to\infty}R(n,d) \label{eq:cap}
  \end{align}
where the maximum in \eqref{eq:rnd} is over all codes with distance $\ge d$.
We have the following result.
\begin{theorem} \cite{BM2010}\label{thm:bounds}
The limit in \eqref{eq:cap} exists, and
   \begin{equation}\label{eq:bounds}
   \sC(d)=\begin{cases} 1 &\text{if } d=O(n)\\
     1-\epsilon &\text{if } d=\Theta(n^{1+\epsilon}), \;0<\epsilon<1\\
     0  &\text{if } d=\Theta(n^2).\end{cases}
  \end{equation}
Moreover,
   \begin{equation*}
       R(n,d)=\begin{cases}O(\log^{-1} n) &\text{if } d=\Theta(n^2)\\
    1-O(\log^{-1} n) &\text {if } d=O(n).
    \end{cases}
  \end{equation*}
 \end{theorem}
\vspace*{.05in} We remark \cite{BM2010}
that the equality $\sC(d) = 1 - \epsilon$ holds under a
slightly weaker condition, namely, $d = n^{1+\epsilon}\alpha(n)$, where $\alpha(n)$
 grows slower than any positive power of $n.$

Equation \eqref{eq:bounds} suggests the following definition. Let us say
  that an infinite family of codes {\em scales optimally} if there exists
  $\epsilon\in (0,1)$ such that, for any positive $\alpha,\beta$, all codes
  of the family of length $n$ larger than some $n_0$, have rate  at least
  $1-\epsilon-\beta$ and minimum distance $\Omega(n^{1+\epsilon-\alpha})$.

The proof of Theorem \ref{thm:bounds} relied on near-isometric embeddings
of $\cX_n$ into other metric spaces that provide insights into the
asymptotic size of codes.
We also showed \cite{BM2010} that there exists a family of rank modulation codes that
correct a constant number of errors and have size within a constant factor of
the upper (sphere packing) bound.

Regarding the problem of explicit constructions, apart from a
construction in \cite{jia10a} of codes that correct one Kendall error,
no other code families for the Kendall distance are presently
known. Addressing this issue, we provide several general constructions
of codes that cover a broad range of parameters in terms of the code
length $n$ and the number of correctable errors.  We present
constructions of rank modulation codes that correct a given number of
errors as well as several asymptotic results that cover the entire
range of possibilities for the scaling of the number of errors with
the code's length. 
Sect.~\ref{sect:pp} we present a construction of low-rate rank
modulation codes that form subcodes of Reed-Solomon codes, and can be
decoded using their decoding algorithms. In Sect.~\ref{sect:h} we
present another construction that gives rank modulation codes capable
of correcting errors whose multiplicity can be anywhere from a
constant to $O(n^{1+\epsilon}), 0<\epsilon<\half,$ although the code
rate is below the optimal rate of \eqref{eq:bounds}. Relying on this
construction, we also show that there exist sequences of rank
modulation codes derived from binary codes whose parameters exhibit
the same scaling rate as \eqref{eq:bounds} for any $0<\epsilon<1.$
Moreover, we show that almost all linear binary codes can be used to
construct rank modulation codes with this optimal trade-off. Finally,
we present a third construction of rank modulation codes from codes
in the Hamming space that correct
a large number of errors. If the number of errors grows as
$\Theta(n^2),$ then the rate of the codes obtained from binary
codes using this construction attains the optimal scaling of $O(\log^{-1} n).$
Generalizing this construction to start from nonbinary codes,
we design families of rank modulation codes that
scale optimally (in the sense of the above definition)
for all values of $\epsilon$, $0<\epsilon <1$.

Finally, Sect.~\ref{sect:examples} contains some examples of codes obtained
using the new constructions proposed here.

Our constructions rely on codes that correct conventional (Hamming)
errors, converting them into Kendall-error-correcting codes. For this reason,
the proposed
methods can be applied to most families of codes designed for the Hamming
distance, thereby drawing on the rich variety of available constructions
with their simple encoding and decoding algorithms.

\section{Construction I:
Rank modulation codes from permutation polynomials}\label{sect:pp}

Our first construction of rank modulation codes is algebraic in nature.
Let $q = p^m$ for some prime $p$ and let $\ff_q=\{\alpha_0,\alpha_1,\dots,\alpha_{q-1}\}$
be the finite field of $q$ elements.
A polynomial $g(x)\in \ff_q[x]$
is called a {\em permutation polynomial} if the values $g(a)$
are distinct for distinct values of $a\in \ff_q$  \cite[Ch.~7]{LN1983}.

Consider the evaluation map $f\mapsto (f(\alpha_0),\dots,f(\alpha_{q-1}))$
which sends permutation polynomials to permutations of $n$ elements.
Evaluations of permutation polynomials of degree $\le k$ form a subset of
a $q$-ary Reed-Solomon code of dimension $k+1$. Reed-Solomon codes are a family
of error-correcting codes in the Hamming space with a number of desirable properties
including efficient decoding. For an introduction to them see \cite[Ch.~10]{MS1977}.

At the same time, evaluating the size of a rank modulation code
constructed in this way is a difficult problem because it is hard
to compute the number of permutation polynomials of a given degree.
In this section we formalize a strategy for constructing codes along these lines.
This does not result in very good rank modulation codes;
in fact, our later combinatorial constructions will be better in terms
of the size of the codes with given error-correcting capabilities.
Nonetheless, the construction involves some interesting observations
which is why we decided to include it.

A polynomial over $\ff_q$ is called {\em linearized of degree
$\nu$} if it has the form
  $$
\mathcal{L}(x) = \sum_{i=0}^{\nu} a_i x^{p^i}
$$
Note that a linearized polynomial of degree $\nu$ has degree
$p^{\nu}$ when viewed as a standard polynomial.
\remove{
Consider a linearized polynomial $ \mathcal{L}(x) = \sum_{i} a_i x^{p^i}$
over the finite field $\ff_q.$ Suppose $\mathcal{L}(x)$ has degree less than or equal to
$n-2t-1,$ for some integer $t.$ So if $\nu = \lfloor\log_p (n-2t) \rfloor$, then,
$$
\mathcal{L}(x) = \sum_{i=0}^{\nu} a_i x^{p^i}.
$$}

\begin{lemma}\label{lem:linperm}
The number of linearized polynomials over $\ff_q$ of degree less than or equal to $\nu$
that are permutation polynomials in $\ff_q$ is at least
$$\Big(1-\frac{1}{p-1}+\frac{1}{q(p-1)}\Big)q^{\nu+1}\geq q^{\nu}.$$
\end{lemma}
\begin{IEEEproof}
The polynomial $\mathcal{L}(x)$ acts on $\ff_q$ as a linear homomorphism.
It is injective if and only if it has a trivial kernel, in other words if
the only root of $\mathcal{L}(x)$ in  $\ff_q$ is $0.$
Hence, $\mathcal{L}(x)$ is a permutation polynomial
if and only if the only root of $\mathcal{L}(x)$ in  $\ff_q$ is $0$.

The total number of linearized polynomials of degree up to $\nu$
is  $q^{\nu+1}.$ We are going to prove that at least a $(1-\frac{1}{p-1}+\frac{1}{q(p-1)})$
proportion
of them are permutation polynomials.
To show this, choose the coefficients $a_i, 0\le i\le \nu$ of
$\mathcal{L}(x)= \sum_{i=0}^{\nu} a_i x^{p^i}$
uniformly and randomly from $\ff_q.$
For a fixed $\alpha \in \ff_q^\ast,$ the probability
that $\mathcal{L}(\alpha) =0$ is $1/q.$ Furthermore, the set of roots
of a linearized polynomial is an $\ff_p$-vector space \cite[p.119]{MS1977},
hence the set of non-zero roots is a multiple of $p-1$.
The number of $1$-dimensional subspaces of $\ff_q$ over $\ff_p$ is
$\frac{q-1}{p-1}.$
The probability
that one of these sets is included in the set of roots of $\mathcal{L}(x)$
is, from the union bound,
$$
\Pr(\exists \alpha \in \ff_q^\ast: \mathcal{L}(\alpha) =0 )
\le \frac{q-1}{p-1}\cdot\frac 1q.
$$
Hence, the probability that $\mathcal{L}(x)$ is a permutation polynomial
is greater than or equal to $1-\frac{q-1}{q(p-1)}.$
\end{IEEEproof}

\subsection{Code construction}
We use linearized permutation polynomials of $\ff_q$
to construct codes in the space $\cX_n.$
Note that a linearized polynomial $\mathcal{L}(x)$ always maps
zero to zero, so that when it is a permutation
polynomial it can be considered to be a permutation of the elements of
$\ff_q$ and also of the elements of $\ff_q^\ast$.
Let $t$ be a positive integer and let
$\nu=\lfloor \log_p(n-2t-1)\rfloor.$
Let $\cP_t$ be the set of all linearized
polynomials of degree $\le \nu$ that permute $\ff_q.$
Set $n=q-1$ and define the set $A\subset \ff_q^n$
  $$
    A=\{(\cL(a), a\in \ff_q^\ast), \;\cL\in \cP_t\}
  $$
to be the set of vectors obtained by evaluating the polynomials in $\cP_t$
at the points of $\ff_q^\ast.$
Form a code $\cC_\tau$ by writing the vectors in $A$ as permutations
(for that, we fix some bijection
between $[n]$ and $\ff_q^\ast,$ which will be implicit in the subsequent discussion).
We can have $n=q$ rather than $n=q-1$ if desired: for that we add the zero field element
in the first position of the $(q-1)$-tuples of $A,$ and the construction below
readily extends.

\vspace*{.05in}
The idea behind the construction is quite simple: the set $A$ is a subset of
a Reed-Solomon code that corrects $t$ Hamming errors. Every Kendall
error is a transposition, and as such, affects at most two coordinates of
the codeword of $\cC_\tau.$
Therefore the code $\cC_\tau$ can correct up to $t/2$ errors.
By handling Kendall errors more carefully, we can actually correct up to
$t$ errors. The main result of this part of our work is given by the following statement.
\begin{theorem}\label{thm:RS} The code $\cC_\tau$ has length $n=q-1$ and size at least
$q^{\lfloor\log_p (n-2t-1) \rfloor}.$ It
corrects all patterns of up to $t$ Kendall
errors in the rank modulation scheme under a decoding algorithm of complexity
polynomial in $n$.
\end{theorem}
\begin{IEEEproof}
It is clear that $|\cC_\tau| = |A|$, and from Lemma \ref{lem:linperm}
$|A| \ge q^{\lfloor\log_p (n-2t-1)\rfloor}.$

Let $\sigma = (a_1, a_2, \ldots, a_i, a_{i+1},\ldots, a_n),$ where
$a_j \in \ff_q^\ast, 1\le j \le n,$ be a permutation in $\cX_n$ (with
the implied bijection between $[n]$ and $\ff_q^\ast$) and let $\sigma'
= (a_1, a_2, \ldots, a_{i+1}, a_i, \ldots, a_n)$ be a permutation
obtained from $\sigma$ by one Kendall step (an adjacent
transposition).  We have
   $$
 \sigma - \sigma' = (0, \ldots, 0 , \theta, -\theta,\ldots, 0)
  $$
where $\theta = a_i -a_{i+1} \in \ff_q^\ast.$

 Let
$$
P =\left(
  \begin{array}{cccccc}
    1 & 0 & 0 & \cdot & \cdot & 0 \\
    1 & 1 & 0 & \cdot & \cdot & 0 \\
    1 & 1 & 1 & \cdot & \cdot & 0 \\
    \cdot & \cdot & \cdot & \cdot & \cdot & \cdot \\
    \cdot & \cdot & \cdot & \cdot & \cdot & \cdot \\
    1 & 1 & 1 & \cdot & \cdot & 1 \\
  \end{array}
\right)
$$
be an $n\times n$ matrix.
Note that
  $$
  P(\sigma - \sigma')^T = (0,\ldots,0,\theta, 0,\ldots,0)^T.
  $$
This means that multiplication by the accumulator matrix $P$ converts one adjacent
transposition error into one Hamming error.
Extending this observation, we claim that if $\dist_\tau(\sigma,\pi)\le t$
with $\pi$ being some permutation, and any $t \le \frac{n}{2},$ then
the Hamming weight of the vector $P(\sigma -\pi)^T$ is not more than $t.$ Here
we again take $\sigma$ and $\pi$ to be vectors with elements from $\ff_q^\ast$
with the implied bijection between $[n]$ and $\ff_q^\ast$.

Now let $\mathcal{L}(x)$ be a linearized permutation polynomial and
let $1, \alpha, \alpha^2,\ldots, \alpha^{q-2}$ be the elements of $\ff_q^\ast$
for some choice of the primitive element $\alpha.$
 Let
  $$
\sigma = (\mathcal{L}(1), \mathcal{L}(\alpha),
 \mathcal{L}(\alpha^2), \ldots, \mathcal{L}(\alpha^{q-2})).
  $$
Since $\cL(a+b)=\cL(a)+\cL(b)$, we have
\begin{eqnarray*}
P\sigma^T =
(\mathcal{L}(\beta_0), \mathcal{L}(\beta_1),
 \mathcal{L}(\beta_2), \ldots, \mathcal{L}(\beta_{q-2}))^T
\end{eqnarray*}
where
   $$
    \beta_i=\sum_{j=0}^i \alpha^{j}, \quad i=0,1,\dots,q-2.
   $$
\remove {(\mathcal{L}(0), \mathcal{L}(0+1), \mathcal{L}(0+1+\alpha),
  \ldots, \mathcal{L}(0+1+\ldots+\alpha^{n-2}))^T\\
 &=& (\mathcal{L}(0), \mathcal{L}(\beta_0), \mathcal{L}(\beta_1),
 \mathcal{L}(\beta_2), \ldots, \mathcal{L}(\beta_{n-2}))^T
\end{eqnarray*}}

It is clear that $\beta_i \ne 0$, $0\le i\le n-1$ and also
$\beta_{i_1} \ne \beta_{i_2}$ for $0\le i_1 < i_2\le n-1,$
Therefore, the vector $P\sigma^T$ is a permutation of the elements
of $\ff_q^\ast.$
At the same time, it is
the evaluation vector of a polynomial
of degree $\le n-2t-1.$
We conclude that the set $\{P\sigma^T, \sigma\in A\}$
is a subset of vectors of an (extended) Reed-Solomon
code of length $n,$ dimension $n-2t$ and distance $2t+1.$
 Any $t$ errors in a codeword of such a code
can be corrected by standard RS decoding algorithms in polynomial time.

The following decoding algorithm of the code $\cC_\tau$
corrects any $t$ Kendall errors. Suppose $\sigma\in A$ is read off
from memory as $\sigma_1.$

\vspace*{.1in}{\em Decoding algorithm (Construction I):}
\begin{itemize}
\item
Evaluate $\bfz=P\sigma_1^T$.
\item Use a
Reed-Solomon decoding algorithm to correct up to $t$ Hamming errors in
the vector $\bfz,$ obtaining a vector $\bfy$ (if the Reed-Solomon decoder
returns no results, the algorithm detects more than $t$ errors).
\item Compute $\sigma=P^{-1}\bfy^T,$ i.e.,
  $$
  \sigma_i=y_{i+1}-y_i, \;\;1\le i\le n-1; \;\;\sigma_n=y_n.
  $$
\end{itemize}
The correctness of the algorithm follows from the construction. Namely,
if $\dist_\tau(\sigma,\sigma_1)\le t,$ then $\bfy$ corresponds to a
transformed version of $\sigma,$ i.e., $\bfy=P\sigma^T.$ Then the last step
of the decoder correctly identifies the permutation $\sigma.$
\end{IEEEproof}

Some examples of code parameters arising from this theorem are given in
Sect.~\ref{sect:examples}.

We note an earlier use of permutation polynomials for constructing permutation
codes in \cite{chu04}. At the same time, since the coding problem considered in that
paper relies on the Hamming metric rather than the Kendall tau distance, its results
have no immediate link to the above construction.

\section{Construction II: Rank modulation codes from the Gray map}\label{sect:h}
In this section we present constructions of rank modulation codes using
a weight-preserving embedding of the Kendall space $\cX_n$ into a subset of integer
vectors. To evaluate the error-correcting capability of
the resulting codes, we
further link codes over integers with codes correcting Hamming errors.

\subsection{From permutations to inversion vectors}
We begin with a description of basic properties of the distance $\dist_\tau$ such
as its relation to the number of inversions in the permutation, and
weight-preserving embeddings of $S_n$ into other metric spaces. Their
proofs and a detailed discussion are found for instance in the books
by Comtet \cite{com74} or Knuth \cite[Sect. 5.1.1]{K1973}.

The distance $\dist_\tau$ is a right-invariant metric which
means that
$\dist_\tau(\sigma_1,\sigma_2)=\dist_\tau(\sigma_1\sigma,\sigma_2\sigma)$ for
any $\sigma,\sigma_1,\sigma_2\in S_n$ where the operation is the usual
multiplication of permutations. Therefore, we can define the weight of the permutation
$\sigma$ as its distance to the identity permutation $e=(1,2,\dots,n).$
%

Because of the invariance, the Cayley graph of $S_n$
(i.e., the graph whose vertices are indexed by the permutations and whose edges connect
permutations one Kendall step apart) is regular of degree $n-1.$ At
the same time it is not distance-regular, and so the machinery of
algebraic combinatorics does not apply to the analysis of the code
structure. The diameter of the space $\cX_n$ equals $N\triangleq\binom n2$
and is
realized by pairs of opposite permutations such as $(1,2,3,4)$ and
$(4,3,2,1).$

The main tool to study properties of $\dist_\tau$ is provided by the
inversion vector of the permutation.  An {\it inversion} in a permutation
$\sigma\in S_n$ is a pair $(i,j)$ such that $i>j$ and
$\sigma^{-1}(j)>\sigma^{-1}(i).$
It is easy to see that
$\dist_\tau(\sigma,e)=I(\sigma)$, the total number of inversions in
$\sigma.$ Therefore, for any two permutations $\sigma_1, \sigma_2$ we
have
$\dist_\tau(\sigma_1,\sigma_2)=I(\sigma_2\sigma_1^{-1})=I(\sigma_1\sigma_2^{-1}).$
In other words,
\begin{eqnarray*}
\dist_{\tau}(\sigma,\pi) &=& |\{(i,j)\in [n]^2 : i\ne j, \pi^{-1}(i) > \pi^{-1}(j),\\
&&\sigma^{-1}(i) < \sigma^{-1}(j)\}|.
\end{eqnarray*}

To a permutation $\sigma\in S_n$ we associate an {\em inversion vector}
$ \bfx_\sigma\in \cG_n\triangleq [0,1]\times [0,2] \times\dots\times [0,n-1], $ where
$\bfx_\sigma(i)=|\{j\in [n]: j<i+1, \sigma^{-1}(j)>\sigma^{-1}(i+1)\}|, i=1,\dots,n-1$.
In words, $\bfx_\sigma(i), i=1,\dots,n-1$ is the number of inversions in $\sigma$ in
which $i+1$ is the first element.
For instance, we have

\vspace*{.05in}
       \begin{tabular}{c@{\hspace*{.5in}}c}
         $\sigma$                 & $\bfx_\sigma$\\
      2 1 6 4 3 7 5 9 8 &1 0 1 0 3 1 0 1
      \end{tabular}
\vspace*{.05in}

\nd It is well known that the mapping
from permutations to the space of inversion vectors is bijective, and
any permutation can be easily reconstructed from its inversion
vector\footnote{There is more than one way to count inversions and to define
the inversion vector: for instance, one can define $\bfx_\sigma(i)=
|\{j: j\le i, \sigma(j)>\sigma(i+1)\}|, i=1,\dots,n-1$.
In this case, given $\sigma=(2,1,6,4,3,7,5,9,8)$  we would have
$\bfx_\sigma=(1,0,2,1,2,0,0,1)$.
The definition in the main text is better suited to our needs in that it supports
Lemma \ref{lemma:wim} below.}.
Clearly,
\begin{equation}\label{inversion_1}
    I(\sigma) =  \sum_{i=1}^{n-1}\bfx_{\sigma}(i).
\end{equation}
Denote by $J:\cG_n\to S_n$ the inverse map from $\cG_n$ to $S_n$, so that
$J(\bfx_\sigma)=\sigma.$
The correspondence
between inversion vectors and permutations was used in \cite{jia10a}
to construct rank modulation codes that correct one error.

For the type of errors that we consider below we introduce the
following $\ell_1$ distance function on $\cG_n:$
    \begin{equation}\label{eq:d}
    \dist_1(\bfx,\bfy)=\sum_{i=1}^{n-1}|\bfx(i)-\bfy(i)|,    \qquad(\bfx,\bfy\in \cG_n)
   \end{equation}
   where the computations are performed over the integers, and write
   $\|\bfx\|$ for the corresponding weight function (this is not a
   properly defined norm because $\cG_n$ is not a linear
   space). Recall that $d_\tau(\sigma,\pi)=I(\pi \sigma^{-1})$; hence the
relevance of the $\ell_1$ distance for our problem.
For instance, let $\sigma_1=(2,1,4,3), \sigma_2=(2,3,4,1),$
   then $\bfx_{\sigma_1}=(1,0,1), \bfx_{\sigma_2}=(1,1,1).$ To compute the distance
   $\dist_\tau(\sigma_1,\sigma_2)$ we note that $\sigma_1^{-1}=\sigma_1$ and so
    $$
     I(\sigma_2\sigma_1^{-1})=I((1,4,3,2))=\|(0,1,2)\|=3.
   $$
Observe that the mapping $\sigma\to\bfx_\sigma$
   is a weight-preserving bijection between $\cX_n$ and the set
   $\cG_n$.
At the same time, the above example shows that this mapping is not distance
preserving. Indeed, $d_\tau(\sigma_1,\sigma_2)=3$ while $d_1(\bfx_{\sigma_1},\bfx_{\sigma_2})=1.$
 However, a weaker property pointed out in \cite{jia10a} is true, namely:
\begin{lemma}\label{lemma:wim} Let $\sigma_1,\sigma_2\in S_n,$ then
   \begin{equation}\label{eq:dd}
     \dist_{\tau}(\sigma_1,\sigma_2)\ge \dist_1(\bfx_{\sigma_1},\bfx_{\sigma_2}).
   \end{equation}
\end{lemma}
\begin{proof}
Let $\sigma(m),\sigma(m+1)$ be two adjacent elements in a permutation
$\sigma.$ Let $i=\sigma(m),j=\sigma(m+1)$ and suppose  that $i<j.$
Form a permutation $\sigma'$ which is the same as $\sigma$
except that $\sigma'(m)=j,\sigma'(m+1)=i,$ so that $d_\tau(\sigma,\sigma')=1.$
The count of inversions for which $i$ is the first element is unchanged, while
the same for $j$ has increased by one. We then have
$\bfx_{\sigma'}(k)=\bfx_{\sigma}(k), k\ne j$ and $\bfx_{\sigma'}(j)=\bfx_\sigma(j)+1.$
Thus, $d_1(\bfx_{\sigma'},\bfx_{\sigma})=1,$ and the same conclusion is clearly
true if $i>j$.

Hence, if the Kendall distance between $\sigma_1$ and $\sigma_2$ is $1$ then
the $\ell_1$ distance between the corresponding inversion vectors is also $1$.
Now consider two graphs $G$ and $G'$ with the same vertex set $S_n$. In $G$
there will be an edge between two vertices if and only if the Kendall distance between them
is $1$. On the other hand there will an edge between two vertices in $G'$ if and only if the
$\ell_1$ distance between corresponding inversion vectors is $1$. We have just shown that
the set of edges of $G$ is a subset of the set of edges of $G'$. The Kendall distance between two permutations
is the minimum distance between them in the graph $G$. A similar statement is true for
the $\ell_1$ distance with the graph $G'.$

This proves the lemma.
\end{proof}

We conclude as follows.
\begin{proposition} If there exists a code $\cC$ in $\cG_n$ with $\ell_1$
distance $d$ then the set $\cC_\tau:=\{J(\bfx): \bfx \in \cC\}$ forms
a rank modulation code in $S_n$ of cardinality $|\cC|$
with
Kendall distance at least $d.$
\end{proposition}

\subsection{From inversion vectors to the Hamming space via Gray Map}
We will need the {\em Gray map} which is a mapping $\phi_s$ from the
ordered set of integers $[0,2^s-1]$ to $\{0,1\}^s$ with the property
that the images of two successive integers differ in exactly one bit.
Suppose that $b_{s-1} b_{s-2} \ldots b_0$, $b_i \in \{0,1\}, 0\le
i<s,$ is the binary representation of an integer $u \in
[0,2^s-1]$. Set by definition $b_s=0$ and define
$\phi_s(u)=(g_{s-1},g_{s-2}, \ldots, g_0),$ where
  \begin{equation}\label{eq:gray}
  g_j = (b_j +b_{j+1}) \pmod 2 \qquad(j=0, 1, \ldots s-1)
  \end{equation}
(note that for $s\ge 4$ there are several ways of defining maps from
integers to binary vectors with the required property).

\vspace*{.1in} {\em Example:} The Gray map for the first 10 integers looks as
follows:
{\small $$
\begin{array}{c@{\hspace*{2pt}}}
    0|\\
    1|\\
    2|\\
    3|\\
    4|\\
    5|\\
    6|\\
    7|\\
    8|\\
    9|\\
      \vdots
\end{array} \quad
\begin{array}{*{8}{c@{\hspace*{2pt}}}}
    0&0&0&0&0&0&0&0\\
    0&0&0&0&0&0&0&1\\
    0&0&0&0&0&0&1&0\\
    0&0&0&0&0&0&1&1\\\hline
    0&0&0&0&0&1&0&0\\
    0&0&0&0&0&   1 &0&1\\
    0&0&0&0&0&   1&1&0\\
    0&0&0&0&0&   1&1&1\\\hline
    0&0&0&0&    1&0&0&0\\
    0&0&0&0&   1&0&0&1\\
      &&&\vdots
\end{array}
\quad {\mathbf \longrightarrow}\quad
\begin{array}{*{8}{c@{\hspace*{2pt}}}}
0&0&0&0&0&0&0&0\\
0&0&0&0&0&0&0&1\\
0&0&0&0&0&0&1&1\\
0&0&0&0&0&0&1&0\\\hline
0&0&0&0&0&1&1&0\\
0&0&0&0&0&1&1&1\\
0&0&0&0&0&1&0&1\\
0&0&0&0&0&1&0&0\\\hline
0&0&0&0&1&1&0&0\\
0&0&0&0&1&1&0&1\\
      &&&\vdots
\end{array}
$$}
\nd Note the ``reflective'' nature of the map: the last 2 bits of the second
block of four are a reflection of the last 2 digits of the first block with respect
to the horizontal line; the last 3 bits of the second block of eight
follow a similar rule, and so on. This property, easy to prove from \eqref{eq:gray}, will be
put to use below (see Prop.~\ref{prop:reflect}).

\vspace*{.1in}
Now, for $i=2,\ldots, n,$ let
$$m_i=\lfloor\log {i}\rfloor,$$ and let
   $$
\psi_i:\{0,1\}^{m_i} \to [0,i-1]
   $$
be the inverse Gray map $\psi_i=\phi_i^{-1}.$ Clearly $\psi_i$ is well defined;
it is injective but not surjective for most $i$'s since the size of its domain is only $2^{m_i}$.
\begin{proposition}
Suppose that $\bfx, \bfy \in \{0,1\}^{m_i}.$ Then
    \begin{equation}\label{eq:lh}
  |\psi_i(\bfx) -\psi_i(\bfy)| \ge \dist_H(\bfx, \bfy),
    \end{equation}
where $\dist_H$ denotes the Hamming distance.
\end{proposition}
\begin{IEEEproof}
This follows from the fact that if $u,v$ are two integers such that
$|u-v|=1$, then their Gray images satisfy $d_H(\phi(u),\phi(v))=1.$
If the number are such that $u<v$ and $|u-v|=d,$ then by the triangle
inequality,
  \begin{align*}
   d_H(\phi(u),\phi(v))&\le d_H(\phi(u),\phi(u+1))\\
            &\hspace*{.4in}+\dots+d_H(\phi(v-1),\phi(v))\\
   &=d
  \end{align*}
\end{IEEEproof}

Consider a vector $\bfx=(\bfx_2| \bfx_3|\ldots |\bfx_n),$ where
$\bfx_i \in \{0,1\}^{m_i},$ $i=2,\ldots,n.$
The dimension of $\bfx$ equals $m=\sum_{i=2}^n m_i\approx \log n!,$ or more precisely
\begin{eqnarray*}
m &=& \sum_{j=1}^{m_n -1} (2^{j+1}-2^j) j +
m_n (n+1-2^{m_n})\\
&=& \sum_{j=1}^{m_n -1} j2^j +
m_n (n+1-2^{m_n})\\
&=& (m_n-2)2^{m_n}+2 +
m_n (n+1-2^{m_n})\\
&=& (n+1)m_n -2^{m_n+1} +2.
\end{eqnarray*}
On the first line of this calculation we used the fact that among the numbers $m_i$ there
are exactly $2^{j+1}-2^j$ numbers equal to $j$
for all $j\le n-1,$ namely those with $i=2^j,2^{j}+1,\dots,2^{j+1}-1.$
The remaining $(n+1)-2^{m_n}$ numbers equal $m_n.$

For a vector $\bfx\in \{0,1\}^m$ let
  $$
  \Psi(\bfx)=\Psi(\bfx_2| \bfx_3|\ldots |\bfx_n)=(\psi_2(\bfx_2),\dots,\psi_n(\bfx_n)).
  $$
  \begin{proposition} \label{prop:1h} Let $\bfx, \bfy \in \{0,1\}^m.$ Then
  $$
\dist_{1}(\Psi(\bfx),\Psi(\bfy)) \ge \dist_H(\bfx,\bfy),
  $$
where the distance $\dist_1$ is the $\ell_1$ distance defined in (\ref{eq:d}).
  \end{proposition}
\begin{IEEEproof} Using \eqref{eq:lh}, we obtain
\begin{eqnarray*}
\dist_{1}(\Psi(\bfx),\Psi(\bfy))&
=& \sum_{i=2}^{n}|\psi_i(\bfx_i)-\psi_i(\bfy_i)| \\
&\ge & \sum_{i=2}^{n} \dist_H(\bfx_i,\bfy_i)\\
&=& \dist(\bfx,\bfy).
\end{eqnarray*}
\end{IEEEproof}

\subsection{The code construction: correcting up to $O(n\log n)$ number of errors}
Now we can formulate a general method to construct rank modulation codes.
We  begin with a binary code $\cA$ of length $m$ and cardinality $M$ in the Hamming space.

\vspace*{.1in} {\em Encoding algorithm (Construction II):} \begin{itemize}
\item Given a message $\bfm$ encode it with the code $\cA$. We obtain a vector
$\bfx\in\{0,1\}^m.$
\item Write $\bfx=(\bfx_2|\bfx_3|\dots|\bfx_n),$ where $\bfx_i\in \{0,1\}^{m_i}.$
\item Evaluate $\pi=J(\Psi(\bfx))$
\end{itemize}
This algorithm is of essentially the same complexity as the encoding of the code
$\cA,$ and if this latter code is linear, is easy to implement,
Both $J$ and $\Psi$ are injective, so the cardinality of the resulting code
is $M$. Moreover, each of the two mappings can only increase the distance (namely, see
\eqref{eq:dd} and the previous Proposition).
Summarizing, we have the following statement.
  \begin{theorem}\label{thm:main}
  Let $\cA$ be a binary code of length
  $$
m=(n+1)\lfloor\log n\rfloor-
   2^{\lfloor \log n\rfloor+1}+2,
  $$
cardinality $M$ and Hamming  distance $d$.
 Then the set of permutations
  $$
    \cC_\tau=\big\{\pi\in S_n : \pi=J(\Psi(\bfx)), \bfx\in \cA\big\}
  $$
forms a rank modulation code on $n$ elements of cardinality $M$ with
distance at least $d$ in the Kendall space $\cX_n$.
  \end{theorem}

The resulting rank modulation code $\cC_\tau$ can be decoded to correct
any $t = \lfloor(d-1)/2\rfloor$ Kendall errors if $t$ Hamming errors are correctable
with a decoding algorithm of the binary code $\cA.$
Namely, suppose that $\sigma'$ is the permutation that represents a corrupted
memory state. To recover the data we perform the following steps.

\vspace*{.1in}{\em Decoding algorithm (Construction II):}
\begin{itemize}
\item Construct the inversion vector $\bfx_{\sigma'}.$ Form a new inversion vector
$\bfy$ as follows. For $i=2,\dots,n$,
if $ \bfx_{\sigma'}(i-1) \in [0,i-1]$ is greater than $2^{m_i} -1$ then
put $\bfy_{\sigma'}(i) = 2^{m_i} -1,$ else put $\bfy_{\sigma'}(i)= \bfx_{\sigma'}(i)$.
\item Form a vector $\bfy\in\{0,1\}^m, \bfy=(\bfy_2|\bfy_3|\dots|\bfy_n)$
where $\bfy_i\in\{0,1\}^{m_i}$ is given by $\phi_i(\bfy_{\sigma'}(i)).$
\item Apply the $t$-error-correcting
decoding algorithm of the code $\cA$ to $\bfy.$ If the decoder
returns no result, the algorithm detects more than $t$ errors.
Otherwise suppose that $\bfy$ is decoded as $\bfx.$

\item Output $\sigma=J(\Psi(\bfx)).$
\end{itemize}

\vspace*{.1in}
The correctness of this algorithm is justified as follows.
Suppose $\sigma\in \cC_\tau$ is the original permutation written into the memory,
and $d_\tau(\sigma,\sigma')\le t.$ Let $\bfx_\sigma$ be its inversion vector
and let $\bfx$ be its Gray image, i.e., a vector such that $\Psi(\bfx) =  \bfx_\sigma.$
By Lemma \ref{lemma:wim} and Prop.~\ref{prop:1h} we conclude that $d_H(\bfx,\bfy)\le t,$
and therefore the decoder of the code $\cA$ correctly recovers $\bfx$ from $\bfy$.
Therefore $\sigma'$ will be decoded to $\sigma$ as desired.

{\em Example: }
Consider a $t$-error-correcting primitive BCH code $\cA$ in the binary Hamming
space of length
$m= (n+1)\lfloor\log  n\rfloor -2^{\lfloor\log  n\rfloor+1} +2$ and designed distance $2t+1$ (generally, we will need to shorten the code to get to the desired length $m$).
The cardinality of the code satisfies
$$
M\ge \frac{2^m}{(m+1)^t}.
$$
The previous theorem shows that we can construct a set of $(n,M)$ rank modulation codes that
correct $t$ Kendall errors. Note that, by the sphere packing bound,
the size of any code $\cC\in \cX_n$ that corrects $t$ Kendall errors  satisfies
$|\cC|= O(n!/n^t)$. The rank modulation codes constructed from binary BCH codes
have size $M=\Omega(n!/(\log n!)^t)=\Omega(n!/(n^t\log^t n)).$

Specific examples of code parameters that can be obtained from the above construction
are given in Sect.~\ref{sect:examples}.

\vspace*{.1in}
{\em Remark (Encoding into permutations):} Suppose that the construction in this section is used to
encode binary messages into permutations (i.e., the code $\cA$ in the above encoding
algorithm is an identity map). We obtain an encoding procedure of binary $m$-bit
messages into permutations of $n$ symbols.
This redundancy of this encoding equals $1-m/\log (n!).$
Using the Stirling formula, we have for $n\ge 1$
  $$
   \log  n!\le \log (\sqrt {2\pi n})+n\log  n-
  \Big(n-\frac1{12 n}\Big) \log  e
  $$
(\cite{abr64}, Eq. 6.1.38).
Writing $m\ge (n+1)(\log n-1)-2n+2,$ we can estimate the redundancy
as
  $$
   1-\frac m{\log n!}\le \frac {(3-\log  e)n}{\log n!}, \quad n\ge 2.
  $$
Thus the encoding is asymptotically nonredundant. The redundancy is the
largest when $n$ is a power of 2.
It is less than 10\% for all $n\ge 69,$ less than 7\% for all $n\ge 527,$
etc.

\subsection{Correcting $O(n^{1+\epsilon})$ number of errors, $0<\epsilon<\half$}
Consider now the case when the number of errors $t$ grows with $n$.
Since the binary codes constructed above are of length about $n\log n$, we can obtain
rank modulation codes in $\cX_n$ that correct error patterns of Kendall
weight $t=\Omega(n\log n).$ But in fact more is true.
We need the following proposition.
\begin{proposition}\label{prop:reflect}
Let $\bfx, \bfy \in \{0,1\}^m.$ Then
  $$
\dist_{1}(\Psi(\bfx),\Psi(\bfy)) \ge \frac{n-1}2 \Big(2^{\frac{\dist_H(\bfx,\bfy)}{n-1}}-1\Big).
  $$
\end{proposition}
\begin{IEEEproof}
Assume without loss of generality that $\bfx\ne\bfy.$
We first claim that, for any such $\bfx,\bfy\in \{0,1\}^{m_i}$, the inequality
$\dist_H(\bfx,\bfy)\geq w_i\ge 1$ implies that
  $|\psi_i(x)-\psi_i(y)|\geq 2^{w_i-1}.$ This is
true because of the reflection property of the standard Gray map
as exemplified above.

Now consider vectors $\bfx =(\bfx_2| \bfx_3|\ldots |\bfx_n) , \bfy =(\bfy_2| \bfy_3|\ldots |\bfy_n)$ in $\{0,1\}^{m}$ where
$\bfx_i,\bfy_i\in \{0,1\}^{m_i}, 2\le i\le n.$
Suppose that $\dist_H(\bfx_i,\bfy_i)= w_i$
for all $i,$ and $\sum_{i=2}^{n} w_i =w$ where
$w=\dist_H(\bfx,\bfy).$

Hence,
\begin{align*}
\dist_{1}(\Psi(\bfx)&,\Psi(\bfy)) = \sum_{i=2}^{n}|\psi_i(\bfx_i)-\psi_i(\bfy_i)| \\
&\ge \sum_{i: \,w_i>0}2^{w_i-1}\\
&= \sum_{i=2}^{n}2^{w_i-1} - \sum_{i:\,w_i=0} \frac12
\end{align*}
We do not have control over the number of nonzero $w_i$'s, so let
us take the worst case. We have
  $$
 \sum_{i=2}^{n}\frac1{n-1}2^{w_i}\ge 2^{\sum_{i=2}^{n}\frac{w_i}{n-1}}
=2^{\frac w{n-1}}.
  $$
As for $\sum_{i:\,w_i=0} \frac12,$ use the trivial upper bound $(n-1)/2.$
Together the last two results conclude the proof.
\end{IEEEproof}
We have the following theorem as a result.
\begin{theorem}
  Let $\cC$ and $\cC_\tau$ be the binary and rank modulation codes defined
  in Theorem~\ref{thm:main}. Suppose furthermore that the minimum Hamming
  distance $d$ of the code $\cC$ satisfies $d=\epsilon m,$ where
  $m$ is the blocklength of $\cC$. Then the minimum Kendall distance
  of the code $\cC_\tau$ is  $\Omega(n^{1+\epsilon}).$
\end{theorem}
\begin{IEEEproof}
We have $\log n -1\le\lfloor\log n\rfloor\le\log n.$ Use this in the definition
of $m$ to obtain that $m\ge n(\log  n-3).$ Therefore,
  $d=\epsilon m \ge \epsilon n(\log  n-3).$
 From the previous proposition the
minimum Kendall distance of $\cC_\tau$ is at least
$$
\frac{n-1}2 \big(2^{\epsilon n(\log  n-3)/(n-1)}-1\big) = \Omega(n^{1+\epsilon}).
$$
\end{IEEEproof}

Examples of specific codes that can be constructed from this theorem are again
deferred to Sect.~\ref{sect:examples}.

Let us analyze the asymptotic trade-off between the rate and the distance of the codes.
We begin with an asymptotically good family of binary codes, i.e.,   a
sequence of codes $C_i, i=1,2\dots,$ of increasing length $m$
for which the rate $\log|C_i|/m$ converges to a positive number $R$, and the
relative Hamming distance behaves as $\epsilon m,$ where $0<\epsilon <\half.$
Such families of codes can be efficiently constructed by means of concatenating
several short codes into a longer binary code (e.g., \cite[Ch.~10]{MS1977})
Using this family in the previous theorem, we obtain a family of rank modulation
codes in $S_n$ of Kendall distance that behaves as $\Omega(n^{1+\epsilon}),$
and of rate $R$ (see (\ref{eq:rank_rate})).
The upper limit of $1/2$ on $\epsilon$ is due to the fact  \cite[p.~565]{MS1977} that no binary codes of large size (of positive rate) are capable of correcting a higher proportion of errors.

\subsection{Correcting even more, $O(n^{1+\epsilon})$, errors, $\half\leq\epsilon < 1$}
It is nevertheless possible to extend the above theorem to the case of
$\epsilon \ge\half,$ obtaining rank modulation codes of
distance $\Omega(n^{1+\epsilon})$, $\half\leq\epsilon < 1$ and
positive rate. However, this extension is not direct, and results in an
existential claim as opposed to the
constructive results above. To be precise, one can show that for any
$0\le \varepsilon<1,$ there exist infinite families of binary
$(m,M,d)$ codes $\cC$, with rate $R=1-\epsilon$, such that the associated
rank modulation code $\cC_{\tau}$ for permutations of $[n]$ in
Theorem~\ref{thm:main} has minimum Kendall distance
$\Omega(n^{1+\varepsilon})$.
\begin{theorem}\label{thm:existence}
  For any $0<\epsilon<1$,
  there exist infinite families of binary $(m,M)$ codes $\cC$ such that
  $(1/m)\log M \to 1 -\epsilon >0,$ and
  the associated rank modulation code $\cC_{\tau}$
  constructed in Theorem~\ref{thm:main} has minimum Kendall distance that scales as
  $\Omega(n^{1+\epsilon})$. Moreover all but an exponentially decaying
  fraction of the binary linear codes are such.
\end{theorem}
\vspace*{.05in}
The rank modulation codes described above have asymptotically
optimal trade-off between the rate and the distance.
Therefore, this family of codes achieves the
capacity of rank modulation codes (see \cite[Thm.~3.1]{BM2010}).

To prove the above theorem we need the help of the following lemma.
\begin{lemma}\label{lem:existencel}
Let $0\le \alpha \le 1$ and let $T \subset [m], |T| \ge \alpha m$ be a coordinate subset.
There exists a binary code $\cC$ of length $m$ and any rate $R < \alpha$ such that
the projections of any two codewords
$\bfx, \bfy \in \cC, \bfx \ne \bfy$ on $T$ are distinct. Moreover all but an
exponentially decaying   fraction of binary linear codes of any rate less than
$\alpha$ are such.
\end{lemma}
\begin{IEEEproof}
The proof is a standard application of the probabilistic method.
Construct a random binary code $\cC$ of length $m$ and size $M = 2^{mR}$ randomly
and independently selecting $M$ vectors from $\{0,1\}^m$ with uniform probability.
Denote by $\cE_{\bfx,\bfy}$ the event that two different  vectors $\bfx,\bfy \in \cC$
agree on $T$.
Clearly $\Pr(\cE_{\bfx,\bfy}) = 2^{-\alpha m},$ for all  $\bfx,\bfy \in \cC.$ The event
$\cE_{\bfx,\bfy}$ is dependent on at most $2(M-1)$ other such events. Using the Lov\'{a}sz
Local Lemma \cite{alo00}, all such events can be avoided, i.e.,
$$
\Pr\big(\bigcap_{\bfx,\bfy\in \cC} \bar{\cE}_{\bfx,\bfy}\big)>0,
$$
if
 $$
e2^{-\alpha m} (2M-1) \le 1
 $$
or
 $$
M \le 2^{\alpha m-1}/e +1/2.
 $$
Hence as long as $R< \alpha,$ there exists a code of rate $R$
that contains no pairs of vectors
$\bfx,\bfy$ that agree on $T$. This proves the first part of the lemma.

To prove the claim regarding random linear codes chose a linear code $\cC$ spanned
by the rows of an $mR \times m$ binary matrix $G$ each entry of which
is chosen independently with $P(0)=P(1)=\half.$
The code $\cC$
will not contain two codewords that project identically on $T$ if
the $mR \times |T|$ submatrix of $G$ with columns indexed by $T$ has full rank.
If $mR < |T|$ then a given $mR \times |T|$
sub-matrix of $G$ has full rank
with probability at least $1-5\cdot 2^{-(|T|-mR)^2}$  \cite{ger86}.
Thus if $|T|$ grows at least as $T=mR +\sqrt m$, the proportion of matrices $G$
in which the $(mR\times T)$ submatrix is singular falls exponentially with $m$.
Even if each of these matrices generates a different code, the proportion of
undesirable codes will decline exponentially with $m.$
\end{IEEEproof}

\begin{IEEEproof}[Proof of Thm. \ref{thm:existence}]
Suppose that $\bfx, \bfy \in \{0,1\}^m$ where $m =\sum_{i=2}^{n} m_i$ and
$m_i =\lfloor\log{i}\rfloor$ as above in this section.
Let
$$
\dist_{1}(\Psi(\bfx),\Psi(\bfy))
=\sum_{i=2}^{n}|\psi_i(\bfx_i)-\psi_i(\bfy_i)| \le n^{1+\epsilon}
$$
for some $0 \le \epsilon\le 1.$
Let $0<\beta<1.$ For at least a $1-\beta$ proportion of indices $i$
we can claim that
  $$
|\psi_i(\bfx_i)-\psi_i(\bfy_i)| \le \frac{n^{1+\epsilon}}{\beta(n-1)}.
  $$

On the other hand, if $\bfx_i$ and $\bfy_i$ have the same value in the
first $t_i$ of the $m_i$ coordinates,
then the construction of the Gray map implies that
$|\psi_i(\bfx_i)-\psi_i(\bfy_i)| \ge 2^{m_i -t_i}.$ Hence
for at least a $1-\beta$ fraction of the $i$'s,
$$
2^{m_i -t_i} \le \frac{n^{1+\epsilon}}{\beta(n-1)},
$$
i.e., $t_i \ge m_i - \epsilon \log{n} - \log{\frac{n}{\beta(n-1)}}.$

Therefore, $\bfx$ and $\bfy$ must coincide in a well-defined subset of coordinates
of size
\begin{align*}
\sum_{i=2}^{\lceil(1-\beta)(n-1)\rceil}t_i &\ge  \sum_{i=2}^{\lceil(1-\beta)(n-1)\rceil}
\!\Big( m_i - \epsilon \log{n} - \log{\frac{n}{\beta(n-1)}}\Big)\\
&= \sum_{i=2}^{\lceil(1-\beta)(n-1)\rceil} \lfloor\log{i}\rfloor  \\&\hspace{.4in}-\epsilon(1-\beta)(n-1) \log{n}
- O(n)\\
&= m (1-\epsilon -O(1/\log n)).
\end{align*}
Invoking Lemma~\ref{lem:existencel} now concludes the proof:
indeed, it implies that there exists a binary code of rate at least $1-\epsilon$
where no such pair of vectors $\bfx$ and $\bfy$ exists. The claim about
linear codes also follows immediately.
\end{IEEEproof}

\subsection{Construction III: A quantization map}

In this section we describe another construction of rank modulation
codes from codes in the Hamming space over an alphabet of size $q\ge 2.$
The focus of this construction is on the case when the number of errors is
large, for instance, forms a proportion of $n^2.$

The first result in this section serves as a warm-up for a more involved
construction given later. In the first construction we use binary codes in
a rather simple manner to obtain codes in permutations.
This nevertheless gives codes in $\cX_n$ that correct a large number
of errors. Then we generalize the construction by using codes
over larger alphabets.

\vspace*{.1in}
\subsubsection{Construction IIIA: Rank modulation codes from binary base codes}
Recall our notation $\cG_n$ for the space of inversion vectors
and the map $J:\cG_n\to S_n$ that sends them to permutations.
Let $\cC\in\{0,1\}^{n-1}$ be a binary code that encodes $k$ bits into $n-1$ bits.

\vspace*{.1in}{\em Encoding algorithm (Construction IIIA):} \begin{itemize}
\item Let $\bfm\in\{0,1\}^k$ be a message. Find its encoding $\bfb$ with the code $\cC$.
\item Compute the vector $\bfx=\vartheta(\bfb),$ where
$\vartheta: \{0,1\}^{n-1} \to \cG_n$ is as follows:
  \begin{gather*}
   \bfb=(b_1, b_2, \ldots, b_{n-1}) \stackrel{\vartheta}\mapsto \bfx
= (x_1, \ldots, x_{n-1})\\
      x_i = \begin{cases}
0 \quad \text{ if } b_{i} =0\\
i \quad \text{ if } b_{i} =1
\end{cases}, \quad i=1,\dots, n-1.
      \end{gather*}
\item Find the encoding of $\bfm$ as $\sigma=J(\bfx).$
\end{itemize}

\vspace*{.1in}\begin{theorem}\label{thm:many} Let $\cC(n-1,M,d\ge 2t+1)$
be a code in the binary Hamming
space and let $\cC_\tau\subset S_n$ be the set of permutations obtained from
it using the above encoding algorithm.
Then the code $\cC_\tau\subset S_n$ has cardinality $M$ and
corrects any $r$ Kendall errors where
  $r=t^2/4$ if $t\ge 2$ is even and $r=(t^2-1)/4$ if $t\ge3$ is odd.
\end{theorem}
\begin{IEEEproof}
To prove the claim about error correction, consider the following decoding procedure
of the code $\cC_\tau.$ Let $\pi$ be a permutation read off from memory.

{\em Decoding algorithm (Construction IIIA):} \begin{itemize}
\item Find the inversion vector
$\bfx_\pi=(x_1, \ldots, x_{n-1}).$
\item Form a
vector $\bfy\in\{0,1\}^{n-1}$ by putting
  \begin{equation*}
y_i = \begin{cases}
0 \quad \text{ if } x_i \le \lfloor \nicefrac{i}2 \rfloor\\
1 \quad \text{ if } x_i > \lfloor \nicefrac{i}2 \rfloor.
\end{cases}
  \end{equation*}
\item Decode $\bfy$ with the code $\cC$ to obtain a codevector $\bfc.$
If the decoder returns no result, the algorithm detects more than $t$ errors.
\item Compute the overall decoding result as $J(\vartheta(\bfc)).$
\end{itemize}

\vspace*{.1in}
Let $\sigma$ be the original permutation, let $\bfx_\sigma$ be its inversion
vector, and let $\bfc(\sigma)$ be the corresponding codeword of $\cC$.
The above decoding can go wrong only if the Hamming distance
$\dist_H(\bfc(\sigma),\bfy)> t.$ For this to happen the $\ell_1$ distance between
$\bfx_\pi$ and $\bfx_\sigma$ must be large, in the worst case satisfying
the condition $\dist_1(\bfx_\pi,\bfx_\sigma) >
\sum_{i=1}^{t}\lfloor i/2 \rfloor.$
This gives the claimed result.
\end{IEEEproof}

 From a binary code in Hamming space of rate $R$ that corrects any $\tau n$
errors, the above construction produces a rank modulation code $\cC_\tau$
of size $2^{Rn}$ that is able to correct $\Omega(n^2)$ errors. The rate
of the obtained code equals $\approx R (\log n)^{-1}.$ According to
Theorem \ref{thm:bounds} this scaling is optimal for the multiplicity of
errors considered. Some numerical examples are given in Sect.~\ref{sect:examples}.

\vspace*{.1in}
\subsubsection{Construction IIIB: Rank modulation codes from nonbinary codes}
This construction can be further generalized to obtain codes that are able to
correct a wide range of Kendall errors by observing that the quantization
map employed above is a rather coarse tool which can be refined if we rely on
codes in the $q$-ary Hamming space for $q>2.$ As a result, for any $\epsilon <1$
we will be able
to construct families of rank modulation codes of rate $R= R(\epsilon)>0$
that correct  $\Omega(n^{1+\epsilon})$ Kendall errors.

Let $l>1$ be an integer. Let $Q=\{a_1,a_2,\dots,a_q\}$ be the code alphabet.
Consider a code $\cC$ of length $n'=2(l-1)(q-1)$ over $Q$ and assume that it corrects
any $t$ Hamming errors (i.e., its minimum Hamming distance is at least $2t+1$).
Let $n = (2l+1)(q-1).$ Consider the mapping $\Theta_q:Q^{n-1}\to\cG_n,$
defined as $\Theta_q(\bfb)=(\vartheta_1(b_1),\vartheta_2(b_2),\dots,\vartheta_{n-1}(b_{n-1})),$
$\bfb = (b_1,\dots,b_{n-1}) \in Q^{n-1},$
where
  \begin{gather*}
   \vartheta_i(a_j)=\begin{cases} 0 &\text{if } i<3(q-1)\\
     (2k-1)(j-1) &\text{if }(2k-1)(q-1)\le i\\&\quad<(2k+1)(q-1) \\&\hspace*{.5in}k= 2,3,\dots,l,
  \end{cases}\\
   j=1,2,3,\dots,q.
 \end{gather*}

To construct a rank modulation code $\cC_\tau$ from the code $\cC$ we perform the
following steps.

\vspace*{.1in}
{\em Encoding algorithm (Construction IIIB):}
\begin{itemize}
\item Encode the message $\bfm$ into a codeword $\bfc\in \cC$\\
\item Prepend the vector $\bfc$ with $3(q-1)-1$ symbols $a_1$.\\
\item Map the obtained $(n-1)$-dimensional vector to $S_n$ using the map
$J\circ\Theta_q.$
\end{itemize}

The properties of this construction are summarized in the following
statement.
\begin{theorem} \label{thm:nonbinary}
Let $n'=2(l-1)(q-1), n=(2l+1)(q-1), l\ge 2$.
Let $\cC(n',M,d=2t+1)$ be a code in the $q$-ary Hamming space.
Then the code $\cC_\tau\subset S_n$ described by the above construction
has cardinality $M$ and corrects any $r$
Kendall errors, where
  $$r=(t+1-(q-1)s)(s+1)-1$$ and $s=\lfloor(t+1)/(2(q-1))\rfloor, s\ge 0.$
\end{theorem}
\vspace*{.1in}\begin{IEEEproof}
We generalize the proof of the previous theorem.
Let $\pi$ be the permutation read off from the memory.

\vspace*{.1in}{\em Decoding algorithm (Construction IIIB):}
\begin{itemize}
\item Find the inversion vector $\bfx_\pi=(x_1, \ldots, x_{n-1}).$
\item Form a $q$-ary
vector $\bfy$ by putting
   \begin{gather*}
    y_i=\begin{cases}
        a_1 &\text{if }i<3(q-1)\\
        a_j &\text{if }(2k-1)(q-1)\le i< (2k+1)(q-1)\\&\hspace*{.1in}\text{ and }
            (2k-1)(j-1)-(k-1) \le x_i\\&\hspace*{1.1in}\le (2k-1)(j-1)+k,\\
     &\hspace*{1.7in}k=2,3\dots,l
      \end{cases}
  \end{gather*}
for $i=1,\dots,n-1.$
\item Decode $\bfy'=(y_{3(q-1)},\dots,y_{n-1})$ with the code $\cC$ to obtain a codevector $\bfc.$ If the decoder returns no results, the algorithm detects more
than $t$ errors.
\item Find the decoded permutation as $\sigma=J(\Theta_q(\bfc)).$
\end{itemize}

\vspace*{.1in}
There will be an error in decoding only when $\bfy'$
contains at least $t+1$ Hamming errors. $\bfy'$ contains
coordinates $3(q-1)$ to $n-1$ of $\bfy.$
Suppose that $t_j, 1\le j\le l-1$ is the number of Hamming
errors in
coordinates between $(2j+1)(q-1)$ and $(2j+3)(q-1)$.
We have $\sum_{j=1}^{l-1}t_j\ge t+1$ and
$t_j \le 2(q-1).$
The $\ell_1$ distance between the received and original
inversion vectors equals
\begin{align*}
\sum_{j =1}^{l-1}j t_j &\ge
\mathop{\min_{t_j \le 2(q-1)}}_{\sum_jt_j \ge t+1}\quad \sum_{j =1}^{l-1}j t_j\\
&= 2(q-1) (1+2+\dots+s) \\&+ (t+1 -2(q-1)s)(s+1)\\
&= (q-1)s(s+1) + (t+1 -2(q-1)s)(s+1)\\
&= (t+1 -(q-1)s)(s+1).
\end{align*}
In estimating the minimum in the above calculation we
have used the fact that the smaller-indexed $t_j$'s should be given
the maximum value before the higher-indexed ones are used.

 Therefore if the $\ell_1$ distance between the received and original
inversion vectors is less than or equal to $r$ then decoding $\bfy'$ with the code $\cC$
will recover $\bfx_\sigma.$
Using \eqref{eq:dd} we complete the proof.
\end{IEEEproof}

\vspace*{.05in} {\em Asymptotic analysis:}
For large values of the parameters we obtain
 that the number of errors correctable by $\cC_\tau$ is
$$
  r\approx \frac {t^2}{4q}
$$
or, in other words, $d(\cC_\tau)\approx d^2/8q.$
In particular, if $d = n'\delta$ 
and $q = O(n^{1-\epsilon})$, $0<\epsilon<1$, then we get
$d(\cC_\tau) = \Omega(n^{1+\epsilon}).$
If the code $\cC$ has cardinality $q^{Rn'}$ then $|\cC_\tau|=q^{Rn'}=
q^{R(n-3(q-1))}.$ Using \eqref{eq:rank_rate} yields the value $(1-\epsilon)R$
for the rate of the code $\cC_\tau.$ This is only by a factor of $R$ less than the
optimal scaling rate of \eqref{eq:bounds}. To achieve the optimal asymptotic rate-distance
trade-off
one need to use a $q$-ary code of rate very close to one and non-vanishing
relative distance; moreover $q$ needs to grow with code length $n$ as $n^{1-\epsilon}.$

To show an example, let us take the family of linear codes on Hermitian curves (see e.g.,
\cite[Ch.~10]{bla08}). The codes can be constructed over any alphabet of size
$q=b^2,$ where $b$ is a prime power. Let $u$ be an integer, $ b+1\le u< b^2-b+1.$
The length $n'$, dimension $k$ and Hamming distance $d$ of the codes are as follows:
$$ n'=b^3+1, \;k=(b+1)u-(1/2)b(b-1)+1,\; d\ge n'-(b+1)u.$$
In the next section we will give a few examples of codes with specific parameters. For the
moment, let us look at the scaling order of $R$ and $r$ as functions of
the length of the codes $\cC_\tau$ obtained from the above arguments.
We have $n\approx q b,$ so $q\approx n^{2/3}$, and
  $$
   R=\frac{k}{n'}=\frac{(b+1)u-(1/2)b(b-1)+1}{b^3+1},
  $$
  $$
   \frac{d}{n'}\ge\frac{b^3-(b+1)u}{b^3+1}.
  $$
Let us choose $u=b^2/2,$
which gives $R\approx \frac12\alpha$ and $\delta\approx
\frac12\alpha,$ where $\alpha=1-O(1/b).$ Finally, we obtain that
the rate of the codes $C_\tau$ behaves as
  $$
    \frac{\log q^{R n'}}{\log n!}=\frac23 R(1-o(1))
  $$
and the number of correctable
Kendall errors is $r\approx(1/64)n^{4/3}$, which gives the scaling order mentioned in the
previous paragraph
for $\epsilon=1/3.$

  By taking $u=b^{1+\gamma}$, for $0<\gamma <1$, and by shortening the Hermitian
  code to the length $\lambda (b+1)u$, for $\lambda >1$ arbitrarily close to $1$  we obtain a code with rate arbitrarily close to $1$ with relative minimum
  distance equal to $1-1/\lambda$. This yields asymptotically optimal scaling,
  in the sense defined in section~\ref{sec:intro}, for values of $\epsilon$
  that range in the interval $(0,1/3)$. For values of $\epsilon$ in the
  range $(1/3,1)$, families of codes with optimal scaling can similarly be
  constructed by starting from Algebraic Geometry codes with lengths that
  exceed larger powers of $q$ than $q^{3/2}$, for instance, codes from the
Garcia-Stichtenoth curves or other curves with a large number of rational points.

Another general example can be derived from the family of quadratic residue (QR)
codes \cite{MS1977}.
Let $p$ be a prime, then there exist QR codes over $\ff_\ell$ of length $n'=p,$ cardinality $M=\ell^{(p+1)/2}$ and distance $\ge\sqrt {p}$, where $\ell$ is a prime
that is a quadratic residue modulo $p.$
Using them in Theorem \ref{thm:nonbinary} (after an
appropriate shortening), we obtain rank modulation codes in $S_n,$ where
$n=p+3(\ell-1),$ with cardinality
$M$ and distance $d(\cC_\tau)=\Omega(p/\ell).$
Let us take a sufficiently large prime $p$ and let $\ell$
be a prime and a quadratic residue modulo p. Suppose that
$\ell=\Theta(p^{\frac 12-\alpha})$ for some small $\alpha>0.$
Pairs of primes with the needed properties
can be shown to exist under the assumption that the
generalized Riemann hypothesis is true (see e.g. \cite{LO1977}).
Using the corresponding QR code $\cC$ in Theorem \ref{thm:nonbinary},
we obtain $n=p+3(\ell-1)=\Theta(p),$ $d(\cC_\tau)=\Theta(n^{\frac12+\alpha})$ and
$\log  M=\Theta(\frac n2(\frac12-\alpha)\log n),$ giving the rate $\frac12(\frac12
-\alpha).$ Although this trade-off does not achieve the scaling order of
\eqref{eq:bounds}, it still accounts for a good asymptotic family of codes.

\section{Examples}\label{sect:examples}
Below $\cC_\tau$ refers to the rank modulation code that we are constructing,
$M=|\cC_\tau|,$ and $t$ is the number of Kendall errors that it corrects.
We write the code parameters as a triple $(n,\log  M,d)$ where $d=2t+1.$
In the examples we do not attempt to
optimize the parameters of rank modulation codes; rather, our goal is to show
that there is a large variety of constructions that can be adapted to the needs
of concrete applications. More codes can be constructed from the codes obtained
below by using standard operations such as shortening or lengthening of codes
\cite{jia10a,BM2010}.
Note also that the design distance of rank modulation
codes constructed below may be smaller than their true distance, so all the
values of the distance given below are lower estimates of the actual values.

 From Theorem \ref{thm:RS} we obtain codes with the following parameters.
Let $q=2^l,$  then $n=q-1$ and $\log M\ge l\lfloor \log (q-2t-2)\rfloor$.
For instance, let $l=6,$ then we obtain the triples $(63,30,31),(63,24,47),$
etc. Taking $l=8,$ we obtain for instance the following sets of parameters:
$(255,56,127),(255,48,191).$

Better codes are constructed using Theorem \ref{thm:main}.
Let us take $n=62,$ then $m=253.$ Taking twice shortened BCH codes $\cB_t$
of length $m$, we obtain a range of rank modulation codes according to the designed
distance of $\cB_t.$ In particular, there are  rank modulation codes in
$\cX_{62}$ with the parameters
  $$
 (62,253-8t,2t+1),t=1,2,3,\dots .
  $$
Similarly, taking $n=105,$ we can construct a suite of rank modulation codes from shortened BCH codes of length $m=510,$ obtaining codes $\cC_\tau$ with the
parameters
  $$
  (105,510-9t,2t+1), t=1,2,3,\dots .
  $$
We remark that for the case of $t=1$ better codes were constructed in
\cite{jia10a}. Namely, there exist single-error-correcting codes in $S_n$
of size $M\ge n!/(2n)$. For instance, for $n=62$ this gives $M=2^{277.064}$
as opposed to our $M=2^{245}.$
A Hamming-type upper bound on $M$ has the form
  $$
    M(t)\le \frac{n!} {\sum_{i=0}^t K_n(i)}
  $$
where
  \begin{align*}
   &K_n(0)=1\\& K_n(1)=n-1\\
   &K_n(2)=(n^2-n-2)/2\\
   &K_n(3)=\binom {n+1}3-n
  \end{align*}
(see e.g., \cite[p.15]{K1973} which also gives a general formula for $K_n(i)$
for $i\le n$).
The codes constructed above are not close to this bound (note however that,
except for small $t$,  Hamming-type bounds are usually loose).

Now let us use binary BCH codes in Theorem \ref{thm:many}.
Starting with codes of length $n'=63,255$
we obtain
rank modulation codes with the parameters
$(64,36,13),$$(64,30,19),$$(64,24,25),$$ (64,18,51),$$ (64,16,61),$
$(64,10,85),$ $(256,215,13),$ $(256,207,19),$ $ (256,199,25),$ $(256,191,33),$ etc.
These codes are not so good for a small number of errors, but
become better as their distance increases.

Finally consider examples of codes constructed from Theorem \ref{thm:nonbinary}.
As our seed codes we consider the following possibilities: products
of Reed-Solomon codes and codes on Hermitian curves.

Let us take $\cC=\cA\otimes\cB,$ where $\cA[15,9,7]$
and $\cB[14,3,12]$ are Reed-Solomon codes over $\ff_{16}.$
Then the code $\cC$ has length $n'=14\cdot15 =210,$ (so $l=8$),
cardinality $16^{27}=2^{108}$ and distance $84,$ so $t=41.$
 From Theorem \ref{thm:nonbinary}
we obtain a rank modulation code $\cC_\tau$ with the parameters
$(n=255,\log M = 108, d \ge107).$
Some further sets of parameters for codes of length $n=255$
obtained as we vary $\dim(\cB)$ are as follows:
  $$
\begin{array}{rccccc}
\dim(\cB) &4&5&6&7&8\\
\log M&144&180&216&252&288\\
d&95&79&67&55&49
\end{array}
  $$
The code parameters obtained for $n=255$ are better than the parameters
obtained for the same length in the above examples with  binary BCH
codes, although decoding product RS codes is somewhat more difficult
than decoding BCH codes
On the other hand, relying on product RS codes offers a great deal of
flexibility in terms of the resulting parameters of rank modulation codes.

We have seen above that Hermitian codes account for some of the best
asymptotic code families when used in Theorem \ref{thm:nonbinary}.
They can also be used to obtain good finite-length rank modulation codes.
To give an example, let $\cC$ be a  projective Hermitian code of
length $4097$ over $\ff_{2^8}$.
We have $\dim(\cC)=17a-119, d(\cC)\ge 4097-17a$ for any integer $a$ such that $17\le a
\le 240;$ see \cite[p.~441]{bla08}.
Let us delete any 17 coordinates (puncture the code)
to get a code $\cC'$ with
\begin{align*}
  n'        &=4080=16(q-1),\\
  \dim(\cC')&=\dim(\cC), \\
  d_H(\cC') &\ge n'-17a.
\end{align*}
We have $n=n'+3(q-1)=4845.$ For $a\in\{60,\dots,100\}$ we obtain a suite of rank modulation
codes with the parameters
$(n,7208,6119),(n,7344,6071),\dots,(n,12648,4079).$

\vspace*{.1in}
As a final remark, note that most existing coding schemes for the Hamming space,
binary or not, can be used in one or more of our constructions to produce
rank modulation codes. The decoding complexity of the obtained codes essentially
equals the decoding complexity of decoding the original codes for correcting Hamming errors
 or for low error probability. This includes codes for which the
Hamming distance is not known or not relevant for the decoding performance,
such as LDPC and polar coding schemes. In this case, the performance of rank modulation
schemes should be studied by computer simulations, similarly to the analysis
of the codes used as building elements in the constructions.

\section{Conclusion}
We have constructed a number of large classes of rank modulation codes, associating
them with binary and $q$-ary codes in the Hamming space.
If the latter codes possess efficient decoding algorithms, then the methods
discussed above translate these algorithms to decoding algorithms of rank modulation
codes of essentially the same complexity. Our constructions also afford simple
encoding of the data into permutations which essentially reduces to the encoding
of linear error-correcting codes in the Hamming space. Thus, the existing theory
of error-correcting codes can be used to design practical error-correcting
codes and procedures for the rank modulation scheme.

A direction of research that has not been addressed in the literature
including the present work, is to construct an adequate model of a probabilistic
communication channel that is associated with the Kendall tau distance.
We believe that the underpinnings of the channel model should be related to
the process of charge dissipation of cells in flash memory devices.
Once a reasonably simple probabilistic description of the error process is
formally modelled, the next task will be to examine the performance on that
channel of code families constructed in this work.

\end{document}